\documentclass[preprint, preprintnumbers,amsmath,amssymb,aps,pra]{revtex4-1}
\usepackage{graphicx}
\usepackage{dcolumn}
\usepackage{bm}
\usepackage{slashed}
\usepackage{amssymb}
\usepackage{amsmath}
\usepackage{amsthm}

\newtheorem{theo}{Theorem}
\begin{document}

\title{Entanglement of disjoint blocks in the one dimensional Spin 1 VBS}

\author{Raul A. Santos }
\email{santos@insti.physics.sunysb.edu} 
\author{V. Korepin}
\email{korepin@max2.physics.sunysb.edu}
\affiliation{C.N. Yang Institute for Theoretical Physics,\\
       Stony Brook University,\\
  Stony Brook, NY 11794-3840, USA}

\begin{abstract}
Starting with the valence bond solid (VBS) ground state of the 1D AKLT Hamiltonian, we make a partition of the system in 2 subsystems $A$ and 
$B$, where $A$ is a block of $L$ consecutive spins and $B$ is it's complement. In that setting we compute the partial transpose density 
matrix with respect to $A$, $\rho^{T_A}$. We obtain the spectrum of the transposed density matrix of the VBS pure system. 
Subsequently we define two disjoint blocks, $A$ and $B$ containing $L_A$ and $L_B$ spins respectively, separated by $L$ sites. 
Tracing away the spins which do not belong to $A\cup B$, we find an expression for the reduced density matrix of the $A$ and $B$ 
blocks $\rho(A,B)$. With this expression (in the thermodinamic limit), we compute the entanglement spectrum and other several entanglement 
measures, as the purity $P={\rm tr}(\rho(A,B)^2)$, the negativity $\mathcal{N}$, and the mutual entropy.
\end{abstract}

\maketitle

Entanglement is a fundamental measure of how much quantum effects we can observe and use, and it is the primary resource in quantum 
computation and quantum information processing \cite{BD,L}. Also entanglement plays a role in quantum phase transitions \cite{OAFF,ON}, 
and even it has been experimentally demonstrated that entanglement may affect macroscopic properties of solids \cite{GRAC,V}.
Currently there is considerable interest in quantifying entanglement in various quantum systems. Measures of entanglement, like entanglement
entropy of correlated electrons, fermions in conformal field theory, spin chains, interacting bosons and other models have been studied 
\cite{AEPW,K,VLRK,JK,hlw,cardy}. For a review of entanglement entropy as an area law see \cite{ECP}.
   
An important measure of entanglement is negativity, introduced in \cite{P}. Negativity is a useful quantity to
characterize quantum effects in mixed systems, where the standard mutual information entropy fails to provide a clear separation between 
classical and quantum correlations. Negativity is also useful in the context of quantum information
because it does not change under local manipulations of the system \cite{Vidal}. It is computed from the partial
transpose density matrix with respect to a subsystem $A$, $\rho^{T_A}$ and essentially measures the degree of which $\rho^{T_A}$ fails 
to be positive. Despite the usefulness of negativity, this quantity is usually difficult to compute analytically.

In this paper, we calculate the negativity of blocks in the ground state of the spin chain introduced by Affleck, Kennedy, Lieb, and Tasaki 
(AKLT model) \cite{AKLT}. This state is known as Valence-Bond Solid (VBS). The AKLT model plays an important role in condensed matter physics,
being the first rigorous example of an isotropic spin chain which agrees with the Haldane conjecture \cite{AKLT0}, i.e.  Haldane's  
suggestion that an anti-ferromagnetic Hamiltonian describing half-integer spins is gapless, while for integer spins it has a gap \cite{H}.
AKLT is also central in an specific scheme of quantum computation, namely measurement based quantum computation \cite{VC,BM}.

While this ground state is fourfold degenerate for an open boundary chain of spin 1 at each site, it becomes
unique for a chain consisting of bulk spin-1’s and two spin-1/2's at the boundary \cite{KK}.  
 
An implementation of AKLT in optical lattices was proposed in \cite{GMC}, and the use of AKLT model for universal quantum computation 
was discussed in \cite{VC} and in \cite{WAR}. VBS is also closely related to Laughlin ansatz \cite{L0} and to fractional quantum Hall effect 
\cite{AAH}.

In the first section we quickly review the formulation of the AKLT model and the VBS ground state with its extension to make it unique. We 
also introduce the definition of the density matrix associated with the VBS ground state. In the second section we discuss the case of a 
bipartition of the pure ground state. We re-derive the spectrum of the partial density matrix $\rho_A={\rm tr}_B\rho$ using our simpler approach
obtaining the results already shown in \cite{F}. We also computed the transposed density matrix $\rho^{T_A}$ to illustrate our method. For this 
case we compute the full spectrum, along with the eigenvectors of $\rho^{T_A}$. We also give a value for the negativity in this case, which 
decays to a constant value twice as fast as expected from the correlation functions. In the third section we define two blocks $A$ and $B$, 
separated by $L$ sites. We compute the density matrix of the mixed system $A\cup B$ $\rho(A,B)$, evaluated by tracing out the environmental 
degrees of freedom. We obtain the spectrum of $\rho(A,B)$ and the entanglement spectrum as function of the separation $L$ between the blocks 
and the size of $A$ and $B$. The purity of this system corresponds to the one encountered for maximally mixed states (up to second order 
corrections). In this section we find that the negativity for this system vanish for non adjacent blocks. We also study the case 
of periodic boundary conditions. In the fourth section we obtain the mutual entropy of the system, in the limit of infinity blocks $A$ and $B$.

\medskip

\section{The AKLT model and the VBS state}

The one dimensional AKLT model that we will consider consists of a chain of $N$ spin-$
1$’s in the bulk, and two spin-$1/2$ on the boundary. The location where the spins sit are called sites. We
shall denote by $\vec{S}_k$ the vector of spin-$1$ operators and by
$\vec{s}_b$ spin$-1/2$ operators, where $b = 0, N + 1$. The Hamiltonian is $H=H_{\rm Bulk}+\Pi_{0,1} + \Pi_{N,N+1},$
where the Hamiltonian corresponding to the bulk is given by

\begin{eqnarray}\label{AKLT_bulk}
 H_{\rm Bulk}&=&\sum_{i=1}^{N-1}\frac{1}{6}\left(3\vec{S}_k\cdot\vec{S}_{k+1} + (\vec{S}_k\cdot \vec{S}_{k+1})^2 +2\right),
\end{eqnarray}

\noindent and the sum runs over the lattice sites. The boundary terms $\Pi$ describe interaction of a spin $1/2$
and spin $1$. Each term is a projector on a state with spin $3/2$:

\begin{equation}
 \Pi_{0,1}=\frac{2}{3}(1+\vec{s}_0\cdot\vec{S}_1), \quad \Pi_{N,N +1} =\frac{2}{3}(1 + \vec{S}_N \cdot\vec{s}_{N+1}).
\end{equation}

In order to construct the ground state $|{\rm VBS}\rangle$ of $H$ we can associate two spin $1/2$ variables at each lattice site 
and create the spin $1$ state symmetrizing them. To prevent the formation of spin $2$, we antisymmetrize states between 
different neighbor lattice sites.  Doing this we are sure that this configuration is actually an eigenstate of 
the Hamiltonian, with eigenvalue $0$ (i.e. the projection of $|VBS\rangle$ on the subspace of spin 2-states is zero). Noting that
the Hamiltonian $H$ is positive definite, then we know that this is the ground state.

We can associate a graph to this state, defining dots as spins $1/2$, links as anti-symmetrization, and circles as symmetrization. 
The graph representation of the VBS ground state is then given by

\begin{center}
\begin{figure}[ht!]
 \includegraphics[scale=.7]{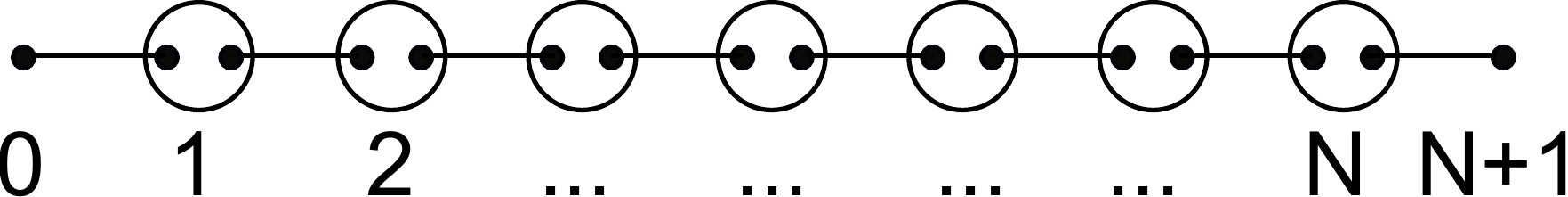}
 \caption{Graphic representation of the 1D VBS state.}
\end{figure}
\end{center}

It is possible to write down a compact expression for this VBS state using bosonic variables. Following \cite{AAH}, we make use of 
the Schwinger boson representation for $SU(2)$ algebra at each site $j$, namely
$S^+_j=a^\dagger_j b_j,\quad S^-_{j}=a_jb^\dagger_j,\quad S^z_{j}=\frac{1}{2}(a_j^\dagger a_j-b_j^\dagger b_j),$
with $[S^z_i,S^\pm_j]=\pm S^\pm_i \delta_{ij}, \quad [S^+_i,S^{-}_j]=+2 S^z_i \delta_{ij}$,
where $a$ and $b$ are two sets of bosonic creation operators, with the usual commutation relations 
$[a_i,a^\dagger_j]=[b_i,b^\dagger_j]=\delta_{ij}$, $[a_i,a_j]=[b_i,b_j]=0$ and correspondingly for $a^\dagger$ and $b^\dagger$. 
This two sets commute in each and every lattice site, i.e. $[a_i,b_j]=[a_i^\dagger,b_j]=0$.
To have a finite dimensional representation of $SU(2)$, it is necessary to impose one more condition on $a$ and $b$, given by
$\frac{1}{2}(a_j^\dagger a_j + b_j^\dagger b_j)= S_j,$  with $S_j$ the value of the spin at site $j$ (in this paper we have $S_j=1$ for 
$j=1..N$ and 1/2 at $j=0,N+1$). In this language the VBS ground state is given by \cite{AAH}

\begin{equation}\label{VBS_state}
 |{\rm VBS}\rangle=\prod_{i=0}^N(a_i^\dagger b_{i+1}^\dagger - a_{i+1}^\dagger b_i^\dagger)|0\rangle.
\end{equation}

\noindent where $|0\rangle=\bigotimes_{sites}|0_a,j\rangle\otimes|0_b,j\rangle$. The state $|0_a,j\rangle$ is defined by $a_j|0_a,j\rangle=0$, 
and it's called the vacuum state for the set of operators $a$. $|0_b,j\rangle$ is defined similarly for the set $b$.
In \cite{KK} the authors prove that this ground state is unique for the Hamiltonian H, then we can construct
the density matrix of the (pure) ground state $\rho=\frac{|{\rm VBS}\rangle\langle {\rm VBS}|}{\langle {\rm VBS}|{\rm VBS}\rangle}.$
This is a one dimensional projector on the $|{\rm VBS}\rangle $ ground state of the Hamiltonian.

\medskip

\section{Density matrix pure state}\label{Neg1}

In order to compute the partial transposed density matrix of the VBS system, we define three subsystems $A$, $B_1$ and $B_2$, where $A$ is a block of $L$ spins 1 and 
$B=B_1\cup B_2$ is it's complement (see Fig. \ref{fig:blocks})

\begin{center}
\begin{figure}[ht!]
 \includegraphics[scale=.7]{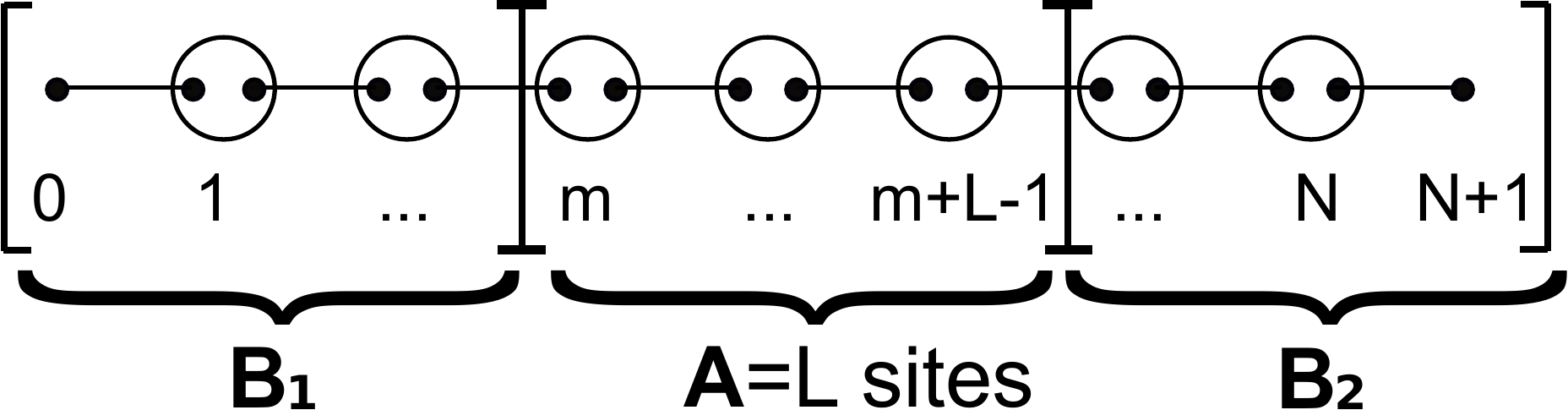}
 \caption{Partition of the 1D chain in three subsystems $A,B_1,B_2$. Subsystem $A$, consisting of $L$ sites. Subsystem $B=B_1\cup B_2$ is 
the complement of $A$.}\label{fig:blocks}
\end{figure}
\end{center}

The partition is defined by (with $1\leq L,m\leq N$) $B_1=\{{\rm sites}\,\,\, i,\,0\leq i \leq m-1\}$, 
$A= \{{\rm sites}\,\, i,\,m\leq i \leq m+L-1\}$, $B_2=\{{\rm sites}\,\, i,\,m+L\leq i \leq N+1\}$

We can split the expression (\ref{VBS_state}) in the corresponding states of the subsystems

\begin{equation}\label{VBS_boundary}
 |{\rm VBS}\rangle=(a_{m-1}^\dagger b_{m}^\dagger - a_{m}^\dagger b_{m-1}^\dagger)(a_K^\dagger b_{K+1}^\dagger - a_{K+1}^\dagger b_K^\dagger)|A,B\rangle.
\end{equation}

\noindent where $K=m+L-1$. $|A\rangle$ and $|B\rangle$ are the VBS states of the $A$  and $B$ subsystems, defined by
$|A\rangle\equiv\left|_m {}^K \right\rangle, \quad |B\rangle\equiv\left|_0 {}^{m-1}\rangle|_{K+1} {}^{N+1}\right\rangle,$
where the states of the form $\left|_I {}^J \right\rangle$ are defined as

\begin{eqnarray}
\left|_I {}^J \right\rangle&\equiv&\prod_{l=I}^{J-1}(a_{l}^\dagger b_{l+1}^\dagger - a_{l+1}^\dagger b_{l}^\dagger)|0\rangle,\\\nonumber
\end{eqnarray}

\noindent respectively. This states describe spins 1 in the bulk (i.e at $l\neq I,J$) and spin 1/2 in the boundary.
To recover the spin 1 at those boundaries sites, we introduce the following notation $(\psi^1_k)^\dagger=a^\dagger_k,$
and $(\psi^2_k)^\dagger=b^\dagger_k$, to have

\begin{eqnarray}\label{definitions}\nonumber
(\psi_m^c)^\dagger(\psi^d_K)^\dagger|A\rangle \equiv |^c A ^d\rangle,\quad
(\psi_{m-1}^c)^\dagger(\psi^d_{K+1})^\dagger|B\rangle \equiv |^c B ^d\rangle \quad (c,d=1,2),
\end{eqnarray}

\noindent then eq. (\ref{VBS_boundary}) becomes
$|{\rm VBS}\rangle=|^2A^1,^1B^2\rangle-|^1A^1,^2B^2\rangle-|^2A^2,^1B^1\rangle+|^1A^2,^2B^1\rangle.$

The four states $|^\sigma A ^\eta\rangle$ belong to the Hilbert space of the block $A$ of length $L$. They span the kernel of 
$H_{\rm Bulk}(A)$ \cite{AKLT}, but they are not orthogonal to each other.
We make use of the classical variable method, introduced in \cite{AAH} (see appendix \ref{appa}), to prove

\begin{eqnarray}
||^\sigma A ^\nu||^2=\frac{1}{4}\left(1-(-1)^{\sigma+\nu}\left(-\frac{1}{3}\right)^L\right),\quad
\langle ^\sigma A ^\nu|^\nu A ^\sigma\rangle=-\frac{1}{2}\left(-\frac{1}{3}\right)^L ~\mbox{for} ~ \sigma\neq\nu,\label{overlap}
\end{eqnarray}

\noindent valid for $L\geq 1$ ($L \in \mathbb{N}$). The norm $||u||^2$ is defined as usual $||u||^2=\langle u|u\rangle$. All other 
combinations vanish. We can perform a rotation of this basis in order to make the overlap (\ref{overlap}) vanish. The new basis is defined by

\begin{eqnarray}\nonumber
|A_0\rangle\equiv\frac{i}{\sqrt{2}}(|^1A^1\rangle+|^2A^2\rangle),\quad\nonumber 
|A_1\rangle\equiv\frac{1}{\sqrt{2}}(|^1A^2\rangle+|^2A^1\rangle),\\\nonumber
|A_2\rangle\equiv\frac{-i}{\sqrt{2}}(|^1A^2\rangle-|^2A^1\rangle).\quad
|A_0\rangle\equiv\frac{1}{\sqrt{2}}(|^1A^1\rangle+|^2A^2\rangle),\label{firstbasis}
\end{eqnarray}

 In this basis (here $\mu,\nu=0..3$) the norm is given by

\begin{eqnarray}\label{norm}
 \langle A_\mu|A_\nu\rangle=\frac{1}{4}\left(1-s_\mu\left(-\frac{1}{3}\right)^L\right)\delta_{\mu\nu},
\end{eqnarray}

\noindent where $s_\mu=(-1,-1,3,-1)$. This four different eigenstates of the bulk Hamiltonian corresponding to the block $A$, can be labeled by the
Bell pair that is formed between the spins 1/2 at the boundary.
The boundary operators $a,b$ which act on the subspace $B$, also organize themselves in irreducible representations, with the only condition
that adjacent boundary operators acting on $A$ and $B$ cannot create a state of spin 2, as required by the VBS ground state symmetry. We define 
the following operators for further simplicity (here, sum over dummy variables $c=1,2$ and $d=1,2$ is assumed)

\begin{eqnarray}\label{B_op}
{T_\mu}^\dagger(i,j)&=&{\psi_i^c}^\dagger(\sigma_\mu)_{cd}{\psi_j^d}^\dagger, \quad (\mu=0..3),
\end{eqnarray}

\noindent with $\sigma_\mu=(i{\rm I},\sigma_1,\sigma_2,\sigma_3)$ being ${\rm I}$ the $2\times 2$ identity matrix and $\sigma_i$ are the Pauli matrices.

The operators  $T_\mu$ keep explicit the symmetry between the operators $a$ and $b$, remaining unchanged (up to phase factors) when we interchange the operators 
$a^\dagger$ and $b^\dagger$. This operation corresponds to take the transpose of $\sigma_\mu$, so $(\sigma_\mu)^T=\sigma_\mu$, for $\mu=0,1,3$
and $(\sigma_2)^T=-\sigma_2$. As the set of operators $T_\mu$ acting on the outer edges of an state is just a linear combination of the states
defined in (\ref{firstbasis}), the states of the form $T_\mu(i,j)\left|_i {}^j \right\rangle$ are a basis for the four dimensional space of 
degenerate ground states of the bulk Hamiltonian (\ref{AKLT_bulk}).

With the introduction of this operators, we can write the identity (see \ref{iden}) (sum convention, with $\mu=0..3$)

\begin{equation}\label{identity}
T_2^{\dagger}(i,i+1)T_2^\dagger(j,j+1)=-\frac{1}{2}{T_\mu}^\dagger(i+1,j){T_\mu}^\dagger(i,j+1),
\end{equation}

Using this identity in eq. (\ref{VBS_boundary}), we find that the VBS state has the decomposition (${T_\nu}^\dagger(m,K)|A\rangle=|A_\nu\rangle$)

\begin{equation}\nonumber
|{\rm VBS}\rangle=-\frac{1}{2}{T_\mu}^\dagger(m-1,K+1){|A_\mu},B\rangle.
\end{equation}

Now we can write the density matrix for the VBS state $\rho=\frac{|{\rm VBS}\rangle\langle VBS|}{\langle{\rm VBS}|{\rm VBS}\rangle}$ 

\begin{equation}\label{density_matrix}
 \rho={T_\mu}^\dagger(m-1,K+1)|A_{\mu},B\rangle\langle A_\alpha,B|T_{\alpha}(m-1,K+1).
\end{equation}

We can define the state $|s\rangle={T_\mu}^\dagger(m-1,K+1){|A_\mu},B\rangle$. In terms of this state, the density matrix (\ref{density_matrix})
takes the form $\rho=|s\rangle\langle s|$, and the eigenvector is clearly $|s\rangle$ with eigenvalue 1. this is natural because so far we have 
just taken another basis to represent $\rho$, which was already a projector onto the VBS ground state.

It's important to note that if we trace the $B$ block in expression (\ref{density_matrix}), we get the partial density matrix
respect to the A subsystem $\rho_A=|A_\mu\rangle\langle A_\mu|$. This expression was already found in \cite{F}. The von Neumman or entanglement
entropy is given by $S_A=-{\rm Tr}(\rho_A\ln\rho_A)=-\lambda_0(L)\ln\lambda_0(L)+3\lambda_1(L)\ln\lambda_1(L)$ and scales with the length of 
the boundary as expected.
The entanglement spectrum for $\rho_A=|A_\mu\rangle\langle A_\mu|$ is $\xi_1=\ln\left(\frac{4}{1+3(-3)^{-L}}\right)$ no degeneracy and 
$\xi_2=\ln\left(\frac{4}{1-(-3)^{-L}}\right)$ with triple degeneracy ($L\neq0$) \cite{F}.

We transpose the elements belonging to the $A$ subspace to obtain the partial transposed matrix $\rho^{T_A}$

\begin{equation}
 \rho^{T_A}={T_\mu}^\dagger(m-1,K+1)|A_{\alpha},B\rangle\langle A_{\mu},B|T_{\alpha}(m-1,K+1)
\end{equation}

\noindent where it is understood that all the operators $T_\mu$ are evaluated at the boundary sites $m-1$ and $K+1$ and $L=K+1-m\geq 1$.
With this explicit form of $\rho^{T_A}$, we can compute the eigenvectors and eigenvalues.

\subsection*{Eigenvectors and negativity of $\rho^{T_A}$}

Defining 

\begin{eqnarray}\nonumber
 \lambda_0(L)\equiv\frac{1}{4}\left(1+3\left(-\frac{1}{3}\right)^L\right)
 \quad\mbox{and} \quad \lambda_1(L)\equiv\frac{1}{4}\left(1-\left(-\frac{1}{3}\right)^L\right)
\end{eqnarray}

\noindent we have:

Eigenvectors of $\rho^{T_A}$ corresponding to $\lambda_1(L)$ (6-fold degeneracy) (here all the $T^\dagger_\mu$ operators act in the sites of
B that are nearest neighbors of the block A, namely $T^\dagger_\mu=T^\dagger_\mu(m-1,K+1)$)

\begin{eqnarray}\nonumber
&|e_1\rangle={T_0}^\dagger|A_0,B\rangle; \quad |e_4\rangle={T_3}^\dagger|A_0,B\rangle-{T_0}^\dagger|A_3,B\rangle,&\\\nonumber
&|e_2\rangle={T_3}^\dagger|A_3,B\rangle; \quad |e_5\rangle={T_1}^\dagger|A_0,B\rangle-{T_0}^\dagger|A_1,B\rangle,&\\
&|e_3\rangle={T_1}^\dagger|A_1,B\rangle; \quad |e_6\rangle={T_1}^\dagger|A_3,B\rangle+{T_3}^\dagger|A_1,B\rangle.&
\end{eqnarray}

Eigenvectors of $\rho^{T_A}$ corresponding to $-\lambda_1(L)$ (3-fold degeneracy)

\begin{eqnarray}\nonumber
|e_7\rangle={T_0}^\dagger|A_3,B\rangle+{T_3}^\dagger|A_0,B\rangle,\\\nonumber
|e_8\rangle={T_0}^\dagger|A_1,B\rangle+{T_1}^\dagger|A_0,B\rangle,\\
|e_9\rangle={T_1}^\dagger|A_3,B\rangle-{T_3}^\dagger|A_1,B\rangle.
\end{eqnarray}

Eigenvectors of $\rho^{T_A}$ corresponding to $\sqrt{\lambda_0(L)\lambda_1(L)}$ (3-fold degeneracy)

\begin{eqnarray}\nonumber
|e_{10}\rangle={T_1}^\dagger|A_2,B\rangle-{T_2}^\dagger|A_1,B\rangle,\\\nonumber
|e_{11}\rangle={T_2}^\dagger|A_0,B\rangle+{T_0}^\dagger|A_2,B\rangle,\\
|e_{12}\rangle={T_3}^\dagger|A_2,B\rangle-{T_2}^\dagger|A_3,B\rangle.
\end{eqnarray}

Eigenvectors of $\rho^{T_A}$ corresponding to $-\sqrt{\lambda_0(L)\lambda_1(L)}$ (3-fold degeneracy)

\begin{eqnarray}\nonumber
|e_{13}\rangle={T_1}^\dagger|A_2,B\rangle+{T_2}^\dagger|A_1,B\rangle,\\\nonumber
|e_{14}\rangle={T_2}^\dagger|A_0,B\rangle-{T_0}^\dagger|A_2,B\rangle,\\
|e_{15}\rangle={T_2}^\dagger|A_3,B\rangle+{T_3}^\dagger|A_2,B\rangle.
\end{eqnarray}

Finally the eigenvector of $\rho^{T_A}$ corresponding to $\lambda_0(L)$ is $|e_{16}\rangle={T_2}^\dagger|A_2,B\rangle.$
The negativity of the partial transposed matrix $\rho^{T_A}$ obtained from the pure VBS state as the sum of negative eigenvalues, 
namely

\begin{eqnarray}\label{negativity_a}\nonumber
 \mathcal{N}&=&\frac{3}{4}\left(1-(-3)^{-L}+\sqrt{\left(1-(-3)^{-L}\right)(1+3(-3)^{-L})}\right),\\
  &=&3(\lambda_1(L)+\sqrt{\lambda_1(L)\lambda_0(L)}),
\end{eqnarray}

\noindent valid for $ L\geq 1$. This expression simplifies in the case $L=1$ to $\mathcal{N}_{L=1}=\frac{3}{4}(1-(-3)^{-1})=1$. In the thermodynamic limit 
($L\rightarrow\infty$) the negativity approach a constant value exponentially fast

\begin{equation}
 \mathcal{N}_{L\rightarrow\infty}=\frac{3}{2}\left(1-\frac{3}{8}\left(-\frac{1}{3}\right)^{2L}\right).
\end{equation}

\subsection*{Special case $L=0$.}

If we evaluate the expression (\ref{negativity_a}) in the case $L=0$, we get zero. This is obvious because in this treatment, $L=0$
means that there is no subspace $A$ to transpose. Nevertheless there is another case which we have not considered so far and is analogous 
to take the block $A$ to be empty. Namely instead of making a partition in three subspaces, we can consider a partition into 2 subspaces.
Let's call this subspaces $B_1$ and $B_2$. The site where we make the partition is $m$ (with $m$ inside the chain). Then the partition is 
given by

\begin{eqnarray}\nonumber
 B_1=\{{\rm sites}\,\, i,0\leq i \leq m\}\quad
 B_2=\{{\rm sites}\,\, i,m+1\leq i \leq N+1\}
\end{eqnarray}

 The density matrix is trivially written as (using the notation introduced in (\ref{definitions}))

\begin{eqnarray}\label{rhoB1B2}
\rho=(\delta^{a}_{c}\delta^{b}_{d}-\delta^{a}_{d}\delta^{b}_{c})|B_1^{a},^{b}B_2\rangle\langle B_1^{c},^{d}B_2|
\end{eqnarray}

The transposed density matrix with respect to $B_2$  is

\begin{eqnarray}
\rho^{T_{B_2}}=(\delta^{a}_{c}\delta^{b}_{d}-\delta^{a}_{d}\delta^{b}_{c})|B_1^{a},^{d}B_2\rangle\langle B_1^{c},^{b}B_2|.
\end{eqnarray}

The eigenvectors of $\rho^{T_{B_2}} $ are, for $\lambda_1^{\{0\}}=\frac{1}{2}$ (three fold degeneracy), 
$|e_k\rangle=T_k^\dagger(m,m+1)|B_1,B_2\rangle$ ($k=1,2,3$), while for the negative eigenvalue $\lambda_2^{\{0\}}=-\frac{1}{2}$ (no degeneracy)
$|e_0\rangle=T_0^\dagger(m,m+1)|B_1,B_2\rangle$. The negativity in this case is $\mathcal{N}=1/2$.

\section{Density matrix for the mixed system of 2 disjoint blocks}\label{N2systems}

\subsection{Open boundary conditions}

So far we have studied the density matrix for the pure VBS system. In this section we want to extend our results to the case of mixed systems.
We will study the mixed system composed of two blocks $A$ and $B$ of length $L_A$ and $L_B$, obtained by tracing away the lattice sites which 
do not belong to these blocks in the VBS ground state. This situation is described in Fig \ref{fig:mutualp}.

\begin{center}
\begin{figure}[ht!]
 \includegraphics[scale=1]{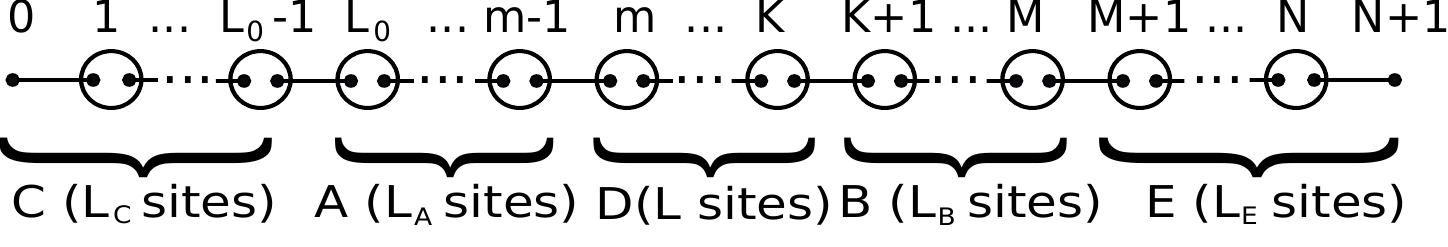}
 \caption{We made a partition of the VBS state in 5 sectors, labeled $A,B,C,D$ and $E$ as shown in the figure. To obtain the density 
matrix for the blocks $A$ and $B$, we trace away the spin variables at the sites inside $C,D$ and $E$.}\label{fig:mutualp}
\end{figure}
\end{center}

To define the blocks, we partition the $N+2$ sites of the chain into five different subsets, $A,B,C,D$ and $E$, of different length.

Given $I,J,K,M,N$ five positive integers ordered as $0< I < J < K < M < N+1$, we define:

\begin{itemize}
 \item Block C = $\{$sites $i$, $0 \leq i \leq I-1\}$, with length $L_C=I$,
 \item Block A = $\{$sites $i$, $I \leq i \leq J-1\}$, with length $L_A=J-I$,
 \item Block D = $\{$sites $i$, $J \leq i \leq K-1\}$, with length $L=K-J$,
 \item Block B = $\{$sites $i$, $K \leq i \leq M-1\}$  with length $L_B=M-K$ and 
 \item Block E = $\{$sites $i$, $M \leq i \leq N+1\}$  with length $L_E=N+2-M$.
\end{itemize}

We are interested in the density matrix for the mixed system of two different blocks. In order to compute the density matrix
we have to trace away the sites outside the corresponding blocks. To do that we use the following results for the bulk states

\begin{eqnarray}\label{orthogonality}
 &\langle D_\mu|D_\nu\rangle=\delta_{\mu\nu}\lambda_\mu(L)\quad \mbox{with } \lambda_\mu=\frac{1}{4}(1+z(L)s_\mu),&\\\nonumber
&z(L)=\left(-\frac{1}{3}\right)^L;  \quad s_\mu=(-1,-1,3,-1), \\
&\langle C,E]T_\mu(I-1,M){T_\nu}^\dagger(I-1,M)\left[C,E\right\rangle=\delta_{\mu\nu}.&
\end{eqnarray}

\subsection*{Density matrix of the blocks C \& E}

The simplest case occur when we trace away the $A, D$ and $B$ blocks. The density matrix for the $C$ and $E$ blocks is

\begin{eqnarray}
\rho_{CE}&=&{\rm Tr}_{ABD}\left\{\frac{|{\rm VBS}\rangle\langle {\rm VBS}|}{\langle {\rm VBS}|{\rm VBS}\rangle}\right\},\\\nonumber
 &=&\lambda_\mu T^\dagger_\mu(I-1,M)|C,E\rangle\langle C,E|T_\mu(I-1,M),\\\nonumber
 &\equiv&\lambda_\mu|[C,E]_\mu\rangle\langle[C,E]_\mu|.
\end{eqnarray}

\noindent the partial transposed density matrix respect to the $E$ system is

\begin{eqnarray}
\rho_{CE}^{T_E}=(\lambda_\mu-(-1)^\mu\frac{\lambda_2-\lambda_1}{2})|[C,E]_\mu\rangle\langle[C,E]_\mu|.
\end{eqnarray}

This is a sum of projector operators (a consequence of (\ref{orthogonality})). The negativity is non vanishing just for the case $L_A+L_B+L=0$, 
when we have $\mathcal{N}=1/2$.

\subsection*{Density Matrix of the blocks A \& B, case $L\geq 1$.}

In this case we compute the density matrix for blocks $A$ and $B$. We obtain this density matrix by tracing away the states 
on the $C,D$ and $E$ subspaces.

\begin{equation}\label{matrix_mix}
 \rho_{AB}={\rm Tr}_{CDE}\left\{\frac{|{\rm VBS}\rangle\langle {\rm VBS}|}{\langle {\rm VBS}|{\rm VBS}\rangle}\right\}.
\end{equation}

Using the identity (\ref{idfour}) (See \ref{iden}), we can write the VBS state as a linear combination of products between the different four 
fold degenerate ground states of the bulk Hamiltonian (\ref{AKLT_bulk}) in the form:

\begin{eqnarray}\nonumber
&|{\rm VBS}\rangle=M_{\mu\nu\rho\sigma} T_\sigma^\dagger(I-1,M)|C,A_\mu,D_\nu,B_\rho,E\rangle,\quad \mbox{with}&\\
&M_{\mu\nu\rho\sigma}=(-1)^\nu(g^{\mu\nu}\delta_{\rho\sigma}+g^{\nu\rho}\delta_{\mu\sigma}-g^{\nu\sigma}\delta_{\mu\rho}+
g^{\nu\alpha}\epsilon_{\mu\alpha\rho\sigma}).&
\end{eqnarray}

\noindent here we have introduced three types of tensors, the Kronecker delta symbol in 4 dimensions $\delta_{\alpha\beta}$, the diagonal tensor 
$g^{\mu\nu}=(-1,+1,+1,+1)$ and the Levi Civita symbol in four dimensions $\epsilon_{\mu\nu\rho\sigma}$, which is a totally antisymmetric tensor, 
with $\epsilon_{\mu\nu\rho\sigma}$= sign of permutation $(\mu,\nu,\rho,\sigma)$ if $(\mu,\nu,\rho,\sigma)$ is a permutation of $(0,1,2,3)$, 
and zero otherwise.

Using this representation of the VBS state, it's easy to write down the density matrix (\ref{matrix_mix}) using the 
{or}\-{tho}\-{go}\-{na}\-{li}\-{ty} of the bulk ground states (\ref{orthogonality}). We find that the density matrix $\rho_{AB}$ is

\begin{equation}\label{dens_matrix_mix}
\rho_{AB}= M_{\mu\nu\rho\sigma} M_{\alpha\nu\beta\sigma}|A_\mu,B_\rho,\rangle\langle A_\alpha,B_\beta|,
\end{equation}

\noindent with the tensor $M_{\mu\nu\rho\sigma}M_{\alpha\nu\beta\sigma}$ given explicitly by  (summation over dummy variables $\nu$ 
and $\sigma$ is assumed)

\begin{eqnarray}\label{dens_matrix_mix2}
 M_{\mu\nu\rho\sigma}M_{\alpha\nu\beta\sigma}=\delta_{\mu\alpha}\delta_{\rho\beta}-z(L)[\delta_{\mu\rho}\delta_{\alpha\beta}-\delta_{\rho\alpha}\delta_{\mu\beta}]S_{\mu\alpha}
+z(L)\epsilon_{\alpha\beta\mu\rho}\left(\frac{S_{\rho\beta}-S_{\mu\alpha}}{2}\right),
\end{eqnarray}

\noindent with $S_{\mu\alpha}=(s_\mu+s_\alpha)/{2}$. We can identify two parts in (\ref{dens_matrix_mix}), the first term which does not depend on 
$z$ and the rest which is linear in $z$. The first term is a projector on the ground states of the bulk of $A$ and $B$, namely

\begin{equation}\label{rho0}
 \rho_0(A,B)=\delta_{\mu\alpha}\delta_{\rho\beta}|A_\mu,B_\rho,\rangle\langle A_\alpha,B_\beta|,
\end{equation}

\noindent while all the other terms, proportional to $z(L)$, have vanishing trace. From this expression, we can compute the entanglement spectrum
associated with $\rho(A,B)$, in  the limit $L,L_1,L_2\gg1$. This density matrix has rank 16, and is exponentially close to a maximally mixed
state. the eigenvalues are (using $x_1=(-3)^{-L_1}$, $x_2=(-3)^{-L_2}$ and $z=(-3)^{-L}$)

\begin{eqnarray}\nonumber
\{\lambda_i\}_{i=1}^{11}=\frac{1-x_1-x_2-z}{16}, &\quad& \{\lambda_i\}_{i=12}^{14}=\frac{1+3x_1+3x_2+3z}{16},\\
\lambda_{15,16}=\frac{1+x_1+x_2+z}{16}&\pm&\frac{1}{8}\sqrt{z^2+(x_1+x_2)(x_1+x_2-z)}.
\end{eqnarray}

\noindent The entanglement spectrum of $\rho(A,B)$ is $\xi_i=-\ln\lambda_i$. Explicitly we have

\begin{eqnarray}\nonumber
\{\xi_i\}_{i=1}^{11}=4\ln2+x_1+x_2+z, &\quad& \{\xi\}_{i=12}^{14}=4\ln2-3x_1-3x_2-3z,\\
\xi_{15,16}=4\ln2 -x_1-x_2-z&\mp&\sqrt{z^2+(x_1+x_2)(x_1+x_2-z)}.
\end{eqnarray}

\noindent The purity, defined as $\gamma={\rm Tr}(\rho^2)$ corresponds in this limit to the purity of a maximally mixed state, up to terms
of second order in $x_1,x_2$ and $z$. We have $\gamma=\frac{1}{16}+\mathcal{O}(2)$. The general results for arbitrary $L,L_1,L_2$ are given 
in appendix \ref{result}.

If we call $\rho_1(A,B)$ to all the linear terms in $z(L)$ on 
(\ref{dens_matrix_mix}), we can write for brevity

\begin{equation}\label{breve}
\rho_{AB}= \rho_0(A,B)+z(L)\rho_1(A,B).
\end{equation}

From the expressions (\ref{dens_matrix_mix}) and (\ref{dens_matrix_mix2}) we can obtain the partial transposed density matrix with respect to 
the $A$ subsystem.

\begin{eqnarray}\label{partialtrans1}
\rho^{T_{A}}_{AB}=\bigg[\delta_{\mu\alpha}\delta_{\rho\beta}
-z(L)\left([\delta_{\alpha\rho}\delta_{\mu\beta}-\delta_{\rho\mu}\delta_{\alpha\beta}]S_{\mu\alpha}
-\epsilon_{\mu\beta\alpha\rho}\left(\frac{S_{\rho\beta}-S_{\mu\alpha}}{2}\right)\right)\bigg]|A_\mu,B_\rho,\rangle\langle A_\alpha,B_\beta|.
\end{eqnarray}

\noindent from where, comparing with equations (\ref{dens_matrix_mix}) and (\ref{dens_matrix_mix2}), we learn that 

\begin{equation}\label{equivalence}
\rho^{T_{A}}_{AB}(z)=\rho_{AB}(-z).
\end{equation}
                                                           
Given this result, and the fact that $\rho_{AB}(-z)$ is also a density matrix (proved in the following theorem), the negativity vanishes for
$L>0$.

\begin{theo}
The negativity of the transposed density matrix $\rho^{T_{A}}_{AB}(z(L))$ is strictly zero for two blocks separated by $L>0$.
\end{theo}
\begin{proof} 
Consider the family of density matrices $\rho_{AB}(z)=\rho_0(A,B)+z(L)\rho_1(A,B)$,  defined in eq. (\ref{breve}). Recalling
that the space of density matrices is convex \cite{Nielsen}, meaning that for two density matrices $\rho_1,\rho_2$, the operator
$\tilde{\rho}=\lambda\rho_1+(1-\lambda)\rho_2$ is also a density matrix for $\lambda\in[0,1]$, we proceed as follows.
We take the first two members of the family $\rho_{AB}(z)$, namely $\rho_{AB}(z_1)$ and $\rho_{AB}(z_2)$ for fixed $z_1=1,z_2=-1/3$
(We can take any pair different $z_1$ and $z_2$, but the greater $z$ is achieved for $z_1=1$, $z_2=-1/3$ (or vice versa)). 
By the convexity of the space of density matrices, $\bar{\rho}=\lambda\rho_{AB}(z_1)+(1-\lambda)\rho_{AB}(z_2)$ is also a density
matrix. Using (\ref{breve}), we write explicitly $\bar{\rho}=\rho_0(A,B)+(\lambda z_1 + (1-\lambda)z_2)\rho_1(A,B)$. We can choose
$\lambda=\frac{1}{4}(1-3(-\frac{1}{3})^L)\in [0,1]$ for $L\geq 1$. Using this $\lambda$, we find 

\begin{equation}
 \bar{\rho}=\rho_0(A,B)-z(L)\rho_1(A,B).
\end{equation}

Then $\bar{\rho}=\rho_{AB}(-z)$ is also a density matrix, for $L\geq 1$. Now, by (\ref{equivalence}), $\rho^{T_{A}}_{AB}(z)$
is also density matrix for $z<1$ $(L>0)$. Then the negativity (sum of negative eigenvalues) of $\rho^{T_{A}}_{AB}(z)$ vanishes.
\end{proof}

\subsection*{Special case $L=0$.}

As in the previous section, we analyze separately the case $L=0$. In this case the block $D$ is not present and we cannot take a trace
over it. 

We can study this scenario using the following identity (see \ref{iden})

\begin{eqnarray}\nonumber\label{ID2}
&T_2^\dagger(i,i+1)T_2^\dagger(j,j+1)T_2^\dagger(k,k+1)=
(-1)^\nu m_{\mu\nu\lambda}{T_\mu}^\dagger(i+1,j){T_\nu}^\dagger(j+1,k){T_\lambda}^\dagger(i,k+1)&\\
&\mbox{where} \quad  m_{\mu\nu\lambda}=\frac{-1}{4}(g_{\mu\nu}\delta_{2\lambda}+g^{2\nu}\delta_{\mu\lambda}-g_{\lambda\nu}\delta_{2\mu}+g^{\nu\alpha}\epsilon_{\mu\alpha 2\lambda})&
\end{eqnarray}

The ${\rm VBS}$ state splits into

\begin{equation}\label{splitVBS}
|{\rm VBS}\rangle=T_2^\dagger(i,i+1)T_2^\dagger(j,j+1)T_2^\dagger(k,k+1)|C,A,B,E\rangle.
\end{equation}
\noindent with $i=L_0-1$, $j=L_1+L_0-1$ and $k=L_0+L_1+L_2-1$ and $L_0$ being the number of sites in $C$, $L_1$ the number of sites in
$A$, $L_2$ the number of sites in $B$. Here the states 
$|C\rangle,|A\rangle,|B\rangle,|E\rangle$ are defined as in the previous section taking $m-1=K$ ($L=0$).

Using the identity (\ref{ID2}), the equation (\ref{splitVBS}) becomes

\begin{eqnarray}\nonumber
|{\rm VBS}\rangle=(g_{\mu\nu}\delta_{2\lambda}+g^{2\nu}\delta_{\mu\lambda}-g_{\lambda\nu}\delta_{2\mu}+g^{\nu\alpha}\epsilon_{\mu\alpha 2\lambda}){T_\lambda}^\dagger(i,k+1)|C,A_\mu,B_\nu,D\rangle
\end{eqnarray}

Now, given that (see appendix \ref{appa})

\begin{equation}
 \langle C,D|{T_\lambda}(i,k+1){T_{\lambda'}}^\dagger(i,k+1)|C,E\rangle=\delta_{\lambda\lambda'},
\end{equation}

\noindent we can write the normalized density matrix ${\rm Tr}_{C,E}\rho\equiv\rho_{AB}$ as

\begin{eqnarray}\nonumber
\rho_{AB}=\bigg(\delta_{\mu\alpha}\delta_{\rho\beta}-[\delta_{\mu\rho}\delta_{\alpha\beta}-\delta_{\rho\alpha}\delta_{\mu\beta}]S_{\mu\alpha}
+\epsilon_{\alpha\beta\mu\rho}\left(\frac{S_{\rho\beta}-S_{\mu\alpha}}{2}\right)\bigg)|A_\mu,B_\nu\rangle\langle A_\alpha,B_\beta|
\end{eqnarray}

This expression is analogous to the (\ref{dens_matrix_mix}), with $z(L)=1$ ($L=0$). From this result, the transposed
density matrix respect to $A$ is now given by 

\begin{eqnarray}\label{rhotrans1}
\rho_{AB}^{T_a}=\bigg(\delta_{\mu\alpha}\delta_{\rho\beta}+[\delta_{\mu\rho}\delta_{\alpha\beta}-\delta_{\rho\alpha}\delta_{\mu\beta}]S_{\mu\alpha}
-\epsilon_{\alpha\beta\mu\rho}\left(\frac{S_{\rho\beta}-S_{\mu\alpha}}{2}\right)\bigg)|A_\mu,B_\nu\rangle\langle A_\alpha,B_\beta|
\end{eqnarray}

Using this expression it is possible to compute the eigenvalues of the transposed density matrix and the negativity. For the results at finite
size of the block $A$ ($L_1$) and block $B$ ($L_2$) see appendix (section \ref{result}). The negativity in the asymptotic limit 
$L_1\rightarrow\infty,L_2\rightarrow\infty$ is given by

\begin{equation}
 \mathcal{N}_{L_1,L_2\rightarrow\infty}=\frac{1}{2}-\frac{3}{4}\left(\left(-\frac{1}{3}\right)^{2L_1}+\left(-\frac{1}{3}\right)^{2L_2}\right).
\end{equation}

\subsection{Periodic Boundary Conditions}\label{PBC}

Using the same technology developed previously, we can also analyze the case of periodic boundary conditions. This state is unique, given that 
the coordination number for each spin is two \cite{KK}. 

In this state, we make a partition in four sectors, labeled by their length as $L_A, L_B, L_C, L_D$, with $L_A+L_B+L_C+L_D=L$ the total length
of the system. We trace away states from sectors that do not belong to $A\cup B$ (See fig. \ref{fig:bc}). 
\begin{center}
\begin{figure}[ht!]
 \includegraphics[scale=.9]{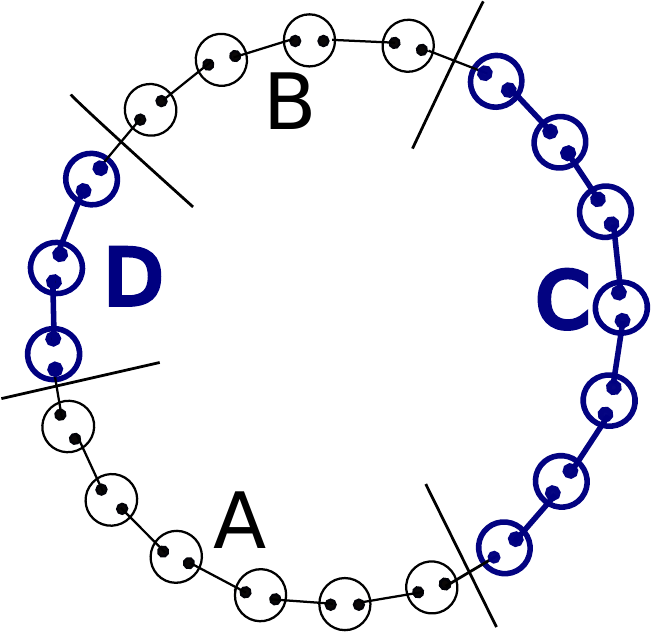}
 \caption{We trace blocks C and D, leaving a reduced density matrix in terms of the states of the A and B blocks.}\label{fig:bc}
\end{figure}
\end{center}

We split the $|{\rm VBS}\rangle$ as in the previous section, but now the difference is that the norm of both states that we trace out, namely
the block C and D, is nontrivial, each one contributing with a factor $\langle C_\alpha| C_\beta'\rangle=\lambda_\alpha(L_C)\delta_{\alpha\beta}$ and
$\langle D_\alpha| D_\beta\rangle=\lambda_\alpha(L_D)\delta_{\alpha\beta}$.

The VBS state can be rewritten as $|{\rm VBS}\rangle=M_{\mu\nu\rho\sigma}|C_\sigma,A_\mu,D_\nu,B_\rho\rangle$, with
$M_{\mu\nu\rho\sigma}=(-1)^\nu(g^{\mu\nu}\delta_{\rho\sigma}+g^{\nu\rho}\delta_{\mu\sigma}-g^{\nu\sigma}\delta_{\mu\rho}+
g^{\nu\alpha}\epsilon_{\mu\alpha\rho\sigma}).$
The reduced density matrix in this case is 

\begin{eqnarray}
\rho_{AB}=M_{\mu\nu\rho\sigma}M_{\alpha\nu\beta\sigma}\lambda_\nu(L_D)\lambda_\sigma(L_C)|A_\mu,B_\rho\rangle\langle A_{\alpha},B_{\beta}|.
\end{eqnarray}

The tensor $W_{\mu\rho\alpha\beta}=M_{\mu\nu\rho\sigma}M_{\alpha\nu\beta\sigma}\lambda_\nu(L_D)\lambda_\sigma(L_C)$ is given explicitly by

\begin{eqnarray}\nonumber
W_{\mu\rho\alpha\beta}=\delta_{\mu\alpha}\delta_{\rho\beta}\Lambda_{\alpha\beta}(z_d,z_c)-\delta_{\alpha\rho}\delta_{\mu\beta}\Gamma_{\alpha\mu}(z_d,z_c)
+\epsilon_{\alpha\beta\mu\rho}R_{\rho\mu\beta\alpha}(z_d,z_c)-\delta_{\mu\rho}\delta_{\alpha\beta}\Gamma_{\mu\alpha}(-z_d,-z_c),
\end{eqnarray}

\noindent where $z_c=z(L_C)$ and $z_d=z(L_D)$. The tensors $\Lambda_{\alpha\beta}(x,y),\Gamma_{\alpha\alpha'}(x,y)$ and $R_{\alpha\beta\alpha'\beta'}(x,y)$
are respectively given by 

\begin{eqnarray}\nonumber
\Lambda_{\alpha\beta}(x,y)=\frac{1+(s_\alpha s_\beta+s_\alpha+s_\beta)xy}{1+3z(L)}&,&
\Gamma_{\alpha\alpha'}(x,y)=\frac{s_\alpha+s_{\alpha'}}{1+3z(L)}\left(xy-\frac{x+y}{2}\right)\\
\mbox{and} \quad R_{\alpha\beta\alpha'\beta'}(x,y)&=&\frac{s_\alpha-s_\beta+s_{\alpha'}-s_{\beta'}}{4+12z(L)}(x-y).
\end{eqnarray}

As with the case studied in section \ref{N2systems}, the partial transposed operator $\rho_{AB}^{T_A}$ is exactly $\rho_{AB}(-z_d,-z_c)$.
With this result, a vanishing negativity of the system is analogous to prove that $\rho_{AB}(-z_d,-z_c)$ is a density matrix. We have

\begin{theo}
The negativity of the transposed density matrix for the system with periodic boundary conditions $\rho^{T_{A}}_{AB}(z_d,z_c)$ is strictly zero 
for $L_C$ \underline{and} $L_D\neq0$.
\end{theo}
\begin{proof} 
Again we proceed as before. The density matrix $\rho_{AB}$ defines a family of operators $\rho_{AB}(z_d,z_c)=\rho_0+z_d\rho_1+z_c\rho_2+z_dz_c\rho_3$,
with 

\begin{flalign}\nonumber
&\rho_0=\frac{\delta_{\mu\alpha}\delta_{\rho\beta}}{1+3z(L)}|A_\mu,B_\rho\rangle\langle A_\alpha,B_\beta|,&\\\nonumber
&\rho_1=\bigg[\frac{s_\alpha+s_\mu}{2+6z(L)}(\delta_{\alpha\rho}\delta_{\mu\beta}-\delta_{\mu\rho}\delta_{\alpha\beta})
+\frac{(s_\alpha-s_\beta+s_\mu-s_\rho)}{4+12z(L)}\epsilon_{\alpha\beta\mu\rho}\bigg]|A_\mu,B_\rho\rangle\langle A_\alpha,B_\beta|,&\\\nonumber
&\rho_2=\bigg[\frac{s_\alpha+s_\mu}{2+6z(L)}(\delta_{\alpha\rho}\delta_{\mu\beta}-\delta_{\mu\rho}\delta_{\alpha\beta})
-\frac{(s_\alpha-s_\beta+s_\mu-s_\rho)}{4+12z(L)}\epsilon_{\alpha\beta\mu\rho}\bigg]|A_\mu,B_\rho\rangle\langle A_\alpha,B_\beta|,&\\\nonumber
&\rho_3=\bigg[\frac{s_\alpha s_\beta+s_\alpha+s_\beta}{1+3z(L)}\delta_{\mu\alpha}\delta_{\rho\beta}
-\frac{(s_\alpha+s_\mu)}{1+3z(L)}(\delta_{\alpha\rho}\delta_{\mu\beta}+\delta_{\mu\rho}\delta_{\alpha\beta})\bigg]|A_\mu,B_\rho\rangle\langle A_\alpha,B_\beta|.&
\end{flalign}

We choose four different members of this family, $\rho_a=\rho_{AB}(1,1)$, $\rho_b=\rho_{AB}(-\frac{1}{3},1)$, $\rho_c=\rho_{AB}(1,-\frac{1}{3})$ and 
$\rho_d=\rho_{AB}(-\frac{1}{3},-\frac{1}{3})$. Recalling that the space of density matrices is convex \cite{Nielsen}, we have that 
$\tilde{\rho}=\alpha\rho_a+\beta\rho_b+\gamma\rho_c+(1-\alpha-\beta-\gamma)\rho_d$ is also a density matrix for $0\leq\alpha,\beta,\gamma\leq 1$.
Choosing $\alpha=\frac{5}{32}+\frac{9}{32}(z_cz_d-z_c-z_d)$, $\beta=\frac{3}{32}+\frac{1}{32}(9z_d-15z_c+9z_cz_d)$ and $\gamma=\frac{3}{32}+\frac{1}{32}(9z_c-15z_d-9z_cz_d)$,
we have

\begin{equation}
 \tilde{\rho}=\rho_0-z_d\rho_1-z_c\rho_2+z_dz_c\rho_3=\rho_{AB}(-z_c,-z_d),
\end{equation}
is also a density matrix for $0\leq\alpha,\beta,\gamma\leq 1$. This condition breaks down when $z_c$ or $z_d$ are equal to 1.

\end{proof}

\section{Mutual Entropy}

Having the explicit expression for the density matrix of two blocks and the proof of vanishing negativity for any separation of the blocks 
greater than zero, a natural question to ask is, Does the mutual entropy vanish in this case, too?. We found surprisingly that the answer is
negative, i.e. there is non zero mutual entropy even when the separation $L$ is greater than zero. The mutual entropy decays exponentially,
as expected from the spin-spin correlations.

From (\ref{dens_matrix_mix}), we can in principle compute the eigenvalues and eigenvectors of $\rho_{AB}$, where $\rho_{AB}$ can be
written as a $16\times16$ matrix. This dimension is fixed by the dimension of the ground state space for each block. Being $V$ the Hilbert space
spanned by the vector $|A_\mu,B_\nu\rangle$ we have

\begin{equation}
 Dim(V)=Dim(\mbox{Ker}A)\times Dim(\mbox{Ker}B)=16
\end{equation}

\noindent where ${\rm Ker}A$ is the Kernel of the bulk Hamiltonian defined on the Hilbert space of the block $A$, and similarly for
${\rm Ker}B$. As we have shown, this spaces are spanned by the states $|A_\mu\rangle,|B_\nu\rangle$ $\mu,\nu=0..3$, states which are orthogonal
but not normalized in our convention. In the thermodynamical limit, when the length of each block goes to infinity
$L_1\rightarrow\infty$,$L_2\rightarrow\infty$, the states $|A_\mu\rangle$ become orthonormal

\begin{equation}
 \langle A_\nu|A_\mu\rangle=\delta_\mu^\nu,\quad\langle B_\mu|B_\nu\rangle=\delta_\mu^\nu.
\end{equation}

\noindent In this limit, the matrix elements of $\rho_{AB}$ are

\begin{eqnarray}\nonumber
\langle A_\mu,B_\nu |\rho_{AB}|A_\alpha,B_\beta\rangle=\delta_{\mu\alpha}\delta_{\nu\beta}+z(L)[\delta_{\mu\nu}\delta_{\alpha\beta}-
\delta_{\nu\alpha}\delta_{\mu\beta}]S_{\mu\alpha}+z(L)\epsilon_{\mu\nu\alpha\beta}\left(S_{\nu\beta}-S_{\mu\alpha}\right),
\end{eqnarray}

The eigenvalues of this matrix are 

\begin{eqnarray}
 \lambda_I(z)&=&\frac{1}{16}(1+3z),\quad \mbox{4-fold degeneracy}\\
\lambda_{II}(z)&=&\frac{1}{16}(1-z),\quad \mbox{12-fold degeneracy}.
\end{eqnarray}

Using the spectral theorem, we find that the entropy of the system described by $\rho_{AB}$, i.e., the entropy of two blocks of infinite 
length separated by $L$ sites is

\begin{eqnarray}
&S[A,B]=-{\rm Tr}(\rho_{AB}\ln(\rho_{AB}))=2\ln2-\frac{3}{4}(1-z)\ln\left(\frac{1-z}{4}\right)-\frac{1}{4}(1+3z)\ln\left(\frac{1+3z}{4}\right).&
\end{eqnarray}

The mutual entropy/information is defined as usual

\begin{equation}
 I(A,B)=S[A]+S[B]-S[A,B],
\end{equation}

\noindent the entropy of a block of length $L$ in the AKLT model was calculated in \cite{F}, and also can be obtained trivially from 
our results of section \ref{Neg1}. In the limit of infinite length, we have

\begin{equation}
 S[A]=S[B]=2\ln2.
\end{equation}

The mutual information is finally

\begin{equation}\label{mutual_info}
 I(A,B)=\frac{3}{4}(1-z)\ln(1-z)+\frac{1}{4}(1+3z)\ln(1+3z),
\end{equation}

\noindent where $z$ was defined before as $z=z(L)=\left(-\frac{1}{3}\right)^L$.

\section{Discussion and Conclusion}

In this paper we have derived the entanglement spectrum of the density matrix of two blocks belonging to the 1D VBS state, corresponding to 
the ground state of the spin 1 AKLT Hamiltonian. The eigenvalues of the density matrix decay exponentially with the length of the blocks and 
their separation to the eigenvalues of a completely mixed state. This decay was expected from the behavior of the correlation functions 
$\langle S^i_0S_L^j\rangle=\frac{4}{3}\left(-\frac{1}{3}\right)^L\delta_{ij}$. The novel result is that in the thermodynamic limit, the 
density matrix $\rho(A,B)$ eq. (\ref{dens_matrix_mix}) is maximally mixed, with all the eigenvalues $\lambda_i=\frac{1}{16}$ for $i=1..16$. The 
density matrix (\ref{dens_matrix_mix}) in this limit is a projector on the different ground states obtained as a tensor products of the four
ground states of $A$ and $B$ blocks $|A_\mu\rangle|B_\nu\rangle$.

The density matrix for this system (\ref{dens_matrix_mix}) is clearly non separable. This can be rigorously proved, using that an 
state is separable iff the quantity ${\rm Tr} (O\rho_{AB})\geq0$ for any Hermitian operator $O$ satisfying ${\rm Tr} (OP\otimes Q)\geq0$, 
where $P$ and $Q$ are projections acting on the Hilbert spaces associated to subsystems $A$ ($\mathcal{H}_A$) and $B$ ($\mathcal{H}_B$). 
If we choose for example $O=r(|A_0,B_0\rangle\langle A_1,B_1|+|A_1,B_1\rangle\langle A_0,B_0|)$, it is easy to see that 
${\rm Tr} (OP\otimes Q)=0$, while ${\rm Tr} (O\rho)\propto r$, then choosing $r$ properly we can make ${\rm Tr} (O\rho)\leq 0$, proving that 
the state is not separable.

The fact that the operator $\rho(A,B)$ describe a state which is non separable, together with the result that the negativity vanishes for 
any separation of the blocks $A$ and $B$ tell us that this state is a bound entangled state \cite{Horodecki2}. In \cite{Horodecki2} the 
authors show that this kind of states can not be distilled by local action to create an useful entanglement for quantum communication tasks 
such as teleportation. However, bound entanglement is still of interest as it can be used to generate a secret quantum key \cite{key}, as 
well as be used to enhance the fidelity of conclusive teleportation using another state \cite{conc-telp}. In the context of many-body 
systems, it has so far been found in thermal states \cite{Ferraro}, XY models \cite{Patane} and in gapless systems \cite{Baghbanzadeh}. 

The mutual entropy of this system, computed in (\ref{mutual_info}) tells us that the work needed to erase all correlations \cite{Groisman} 
between two different blocks in the AKLT ground state decay exponentially to zero in the limit of infinite length. In this sense, all 
correlations between two blocks located infinitely far apart vanish (for non entangled boundary spins) which is expected from the 
thermodinamic limit of a gapped system.

We also studied other different boundary conditions which, altogether with the results found in \cite{SKS}, agree in the limit of infinite 
separation, except in the case when we start with an entangled pair at the spin 1 boundary of a free end AKLT ground state. In that case
as shown in \cite{SKS}, the entanglement reduces to the entanglement of the Bell pair created from the spin 1/2 virtual particles which
remain free after the antisymetrization between different neighbor sites.

Our result is in agreement with the fact that 1D AKLT chains alone are not sufficient for universal quantum computation. This is due to the 
vanishing negativity between two different non adjacent blocks. Still further coupling of many such chains can in principle implement quantum 
computation as shown in \cite{Miyake}.

\medskip

\noindent \textbf{Acknowledgements} R. S. acknowledges the Fulbright-Conicyt Fellowship. V.K. acknowledges the Grant DMS-0905744.

\medskip
\section{Appendix}

\subsection{Classical Variable representation}\label{appa}

A known representation of the boson algebra introduced in the first section is given by
$ a^\dagger_i=u_i,$ $a_i=\frac{\partial}{\partial u_i}$, $b^\dagger_i=v_i$, $b_i=\frac{\partial}{\partial v_i}$.
In this representation the spin operators read

\begin{eqnarray}\label{spin_op}
 S^{+}_i=u_i\frac{\partial}{\partial v_i}, \quad S^{-}_i=v_i\frac{\partial}{\partial u_i},
S^z_i=\frac{1}{2}\left(u_i\frac{\partial}{\partial u_i}-v_i\frac{\partial}{\partial v_i}\right)
\end{eqnarray}

Due to the rotational invariance of the AKLT model, it's useful to choose \cite{AAH}

\begin{eqnarray}
 u_i=e^{i\phi_i/2}\cos\frac{\theta_i}{2}, \quad v_i=e^{-i\phi_i/2}\sin\frac{\theta_i}{2},
\end{eqnarray}

\noindent where $\theta$ and $\phi$ parametrize the unit sphere, with $\theta\in[0,\pi]$ $\theta=0$ being the positive $z$ axis and 
$\phi\in[0,2\pi]$.

The condition $a^\dagger_i a_i+b^\dagger_i b_i=2S$ imposes a restriction on the functions allowed to form spin states, namely

\begin{equation}\label{condition}
 \frac{1}{2}\left(u_i\frac{\partial}{\partial u_i}+v_i\frac{\partial}{\partial v_i}\right)f(u_i,v_i)=Sf(u_i,v_i)
\end{equation}

The solution to (\ref{condition}) is $f(u,v)=\sum_kf_ku^kv^{2S-k}$, a polynomial of degree $2S$ in $u$ and $v$, with $a_k$ an arbitrary constant.
The inner product becomes $\langle g|f\rangle=\int\frac{d\Omega}{2\pi}\bar{g}(u,v)f(u,v)$ where $\Omega$ is the solid angle over the sphere, 
$\bar{g}$ is the complex conjugate of $g$. In the subspace of degree $2S$ polynomials the matrix elements for $S^+$ are

\begin{eqnarray}\nonumber
 \langle g(u,v)|S^{+}f(u,v)\rangle\equiv\langle g(u,v)|u\frac{\partial}{\partial v}f(u,v)\rangle
=\sum_{j,k}\bar{a}_ja_k\int\frac{d\Omega}{4\pi}\bar{u}^j\bar{v}^{2S-j}u\frac{\partial}{\partial v}u^kv^{2S-k}\\
=\sum_{j,k}\bar{a}_ja_k\delta_{k+1,j}B(k+2,2S-k)(2S-k).
\end{eqnarray}

\noindent  with $B(x,y)$ the beta function.
Fo\-llo\-wing \cite{AAH}, we introduce the classical variable representation of $S^+$ as $S^+_{cl}=2(S+1)u\bar{v}.$
The matrix elements for this operator are

\begin{eqnarray}
 \langle g(u,v)|S^{+}_{cl}f(u,v)\rangle\equiv\langle g(u,v)|2(S+1)u\bar{v}f(u,v)\rangle\\\nonumber
=\sum_{j,k}2(S+1)\bar{a}_ja_k\int\frac{d\Omega}{4\pi}\bar{u}^j\bar{v}^{2S-j+1}u^{k+1}v^{2S-k}\\
=\sum_{j,k}\bar{a}_ja_k\delta_{k+1,j}B(k+2,2S-k+1)(2S+2).
\end{eqnarray}

Now, writing the beta function in terms of the gamma function, and using $\Gamma(z+1)=z\Gamma(z)$, we have

\begin{eqnarray}\nonumber
 B(k+2,2S-k+1)=\frac{\Gamma(k+2)\Gamma(2S-k+1)}{\Gamma(2S+3)}
=\frac{(2S-k)}{2S+2}B(k+2,2S-k).
\end{eqnarray}

\noindent then $\langle g(u,v)|S^{+}_{cl}f(u,v)\rangle=\langle g(u,v)|S^{+}f(u,v)\rangle$. The cla\-ssi\-cal representation of the spin 
operators is $S^+_{cl}=(2S+2)u\bar{v}$, $S^-_{cl}=(2S+2)v\bar{u}$, and $S^z_{cl}= (S+1)(u\bar{u}-v\bar{v})$. This expressions provide the 
same matrix elements as the operators (\ref{spin_op}), as can be shown easily from the definitions.

Now using the relations $a=u=e^{i\phi/2}\cos\frac{\theta}{2}, a^\dagger=\bar{u}=e^{-i\phi/2}\cos\frac{\theta}{2}$ and 
$b=v=e^{-i\phi/2}\sin\frac{\theta}{2}, b^\dagger=\bar{v}=e^{i\phi/2}\sin\frac{\theta}{2}$, we can also prove that a similar relation holds for
the overlap between states satisfying (\ref{condition})

\begin{eqnarray}
\frac{\langle g(a,b)|f(a,b)\rangle}{\sqrt{\langle g|g\rangle\langle f|f\rangle}}=\frac{\int\frac{d\Omega}{2\pi}\bar{g}(u,v)f(u,v)}
{\sqrt{\int\frac{d\Omega}{2\pi}|g(u,v)|^2\int\frac{d\Omega}{2\pi}|f(u,v)|^2}}
\end{eqnarray}

The state $|_i{}^{i+L-1}\rangle$ containing $L$ sites fulfills the condition (\ref{condition}) at every lattice point,
except at the boundary sites $i$ and $i+L-1$. The ground states $|A^S_\alpha\rangle$ of the bulk Hamiltonian defined by  
$|A_\mu\rangle={T_\mu}^\dagger(i,j+1)\left|_i{}^{j+1}\right\rangle$, introduced in (\ref{firstbasis}), satisfy the relation (\ref{condition}) 
at each lattice site. 

The norm of the VBS state (\ref{VBS_state}) is in this language
\begin{eqnarray}\nonumber
 \langle {\rm VBS}|{\rm VBS}\rangle&=&\int\left[\prod_{i=0}^{N+1}\frac{d\Omega_i}{4\pi}\right]\prod_{i=0}^N(1-\hat{\Omega}_i\cdot\hat{\Omega}_{i+1})\\
&=&1.
\end{eqnarray}

\noindent with $\hat{\Omega}_i$ being the radial vector on the unit sphere $\hat{\Omega_i}=(\sin\theta_i\cos\phi_i,\sin\theta_i\sin\phi_i,\cos\theta_i).$

The norm of the states $|A_\mu\rangle={T_\mu}^\dagger(i,i+L-1)\left|_i{}^{i+L-1}\right\rangle$ composed by $L$ sites is then

\begin{eqnarray}
\langle A_\mu|A_\nu\rangle&=&\int\prod_{i=1}^{L}\frac{d\Omega_i}{4\pi}(1-\hat{\Omega}_i\cdot\hat{\Omega}_{i+1})\mathcal{T_\mu}^*(1,L)\mathcal{T_\nu}(1,L)
=\frac{1}{4}\left(1+s_\mu\left(-\frac{1}{3}\right)^L\right)\delta_{\mu\nu},
\end{eqnarray}

\noindent with $s_\mu=(-1,-1,3,-1)$. Here we have introduced $\mathcal{T_\mu}$, the classical analog to the operators $T_\mu$, defined as

\begin{equation}
 \mathcal{T}_\mu(i,j)=\varphi_i^a(\sigma_\mu)_{ab}\varphi_j^b, \,\, \mbox{with} \,\,  \varphi_i^{1}=u_i,\,\varphi_i^{2}=v_i.
\end{equation}

\subsection{Identities}\label{iden}

All the identities that we have use in this work can be obtained from the basic identity,
(repeated indices are summed) 

\begin{flalign}\label{iden_gen}
{\psi_i^a}^\dagger{\psi^b_{k}}^\dagger=-\frac{1}{2}(-1)^\mu{T_\mu}^\dagger(i,k)(\sigma_\mu)_{ab},
\end{flalign}

\noindent which can be checked directly by inspection, and makes use of the fact that the matrices $\sigma_\mu$ form a basis of the 
$GL(2,\mathbb{C})$ group. Then for two boundary operators (these operators appear naturally in the boundary between two different blocks in the bulk) 
$\hat{\partial}^\dagger_i={T_2}^\dagger(i,i+1)$ and $\hat{\partial}^\dagger_j={T_2}^\dagger(j,j+1)$, we have

\begin{flalign}\nonumber
\hat{\partial}^\dagger_i\hat{\partial}^\dagger_j={\psi_i^a}^\dagger(\sigma_2)_{ab}{\psi^b}^\dagger_{i+1}{\psi_j^c}^\dagger(\sigma_2)_{cd}{\psi^d}^\dagger_{j+1}
\end{flalign}

\noindent now applying the identity (\ref{iden_gen}) twice and using the fact that $(\sigma_\nu)_{ad}=(\sigma_\nu^T)_{da}$ which is also equal to
$-(-1)^\nu g^{\nu\nu'}(\sigma_{\nu'})_{da}$, we get

\begin{flalign}\label{iden2}\nonumber
\hat{\partial}^\dagger_i\hat{\partial}^\dagger_j=\frac{1}{4}(-1)^{\mu}g^{\nu\nu'}{T_\mu}^\dagger(i+1,j){T_\nu}^\dagger(i,j+1){\rm Tr}(\sigma_2\sigma_\mu\sigma_2\sigma_{\nu'}).
\end{flalign}

\noindent The last term in the above expression is the trace of 4 matrices. To compute it we can use that $\sigma_2\sigma_\mu\sigma_2=(-1)^\mu\sigma_\mu$ 
and that ${\rm Tr}(\sigma_\mu\sigma_\nu)=2g^{\mu\nu}$ obtaining

\begin{equation}
 \hat{\partial}^\dagger_i\hat{\partial}^\dagger_j=-\frac{1}{2}{T_\mu}^\dagger(i+1,j){T_\mu}^\dagger(i,j+1)
\end{equation}

Using this identities it is possible to generate all the identities for any number of boundary operators $\hat{\partial}$. In this 
work, we used the identity for three and four $\hat{\partial}$ operators

\begin{flalign}\label{idthree}
&\hat{\partial}^\dagger_i\hat{\partial}^\dagger_j\hat{\partial}^\dagger_k=-\frac{1}{8}(-1)^\nu g^{\nu\nu'}{T_\mu}^\dagger(i+1,j){T_\nu}^\dagger(j+1.k){T_\lambda}^\dagger(i,k+1)
\times {\rm Tr}(\sigma_\mu\overline{\sigma}_{\nu'}\sigma_2\overline{\sigma}_\lambda)&,\\
&\hat{\partial}^\dagger_i\hat{\partial}^\dagger_j\hat{\partial}^\dagger_k\hat{\partial}^\dagger_l=\frac{(-1)^\nu}{16} g^{\nu\nu'}{T_\mu}^\dagger(i+1,j){T_\nu}^\dagger(j+1,k){T_\rho}^\dagger(k+1,l){T_\lambda}^\dagger(i,l+1)
\times {\rm Tr}(\sigma_\mu\overline{\sigma}_{\nu'}\sigma_\rho\overline{\sigma}_\lambda).&\label{idfour}
\end{flalign}

\noindent here $\sigma_\mu=(i,\sigma_1,\sigma_2,\sigma_3)$ where $\sigma_k$ are the three Pauli matrices.
We also define $\overline{\sigma}_\mu=(-i,\sigma_1,\sigma_2,\sigma_3)$. To compute the traces of the Pauli Matrices, we use the
following tricks.

\begin{eqnarray}\nonumber
\sigma_\mu\overline{\sigma}_\nu+\sigma_\nu\overline{\sigma}_\mu=2\delta_{\mu\nu},\,\,
\frac{\sigma_\mu\overline{\sigma}_\nu-\sigma_\nu\overline{\sigma}_\mu}{2}\equiv\sigma_{\mu\nu},\,\,
\sigma_\mu\overline{\sigma}_\nu=\delta_{\mu\nu}+\sigma_{\mu\nu}
\end{eqnarray}

The $\sigma_{\mu\nu}$ object is a generator of the Lorentz transformations in Euclidean space, so it satisfies the Euclidean Lorentz algebra
 
\begin{equation}
[\sigma_{\mu\nu},\sigma_{\alpha\beta}]=2(\delta_{\nu\alpha}\sigma_{\mu\beta}-\delta_{\nu\beta}\sigma_{\mu\alpha}+\delta_{\mu\beta}\sigma_{\nu\alpha}-\delta_{\mu\alpha}\sigma_{\nu\beta}).
\end{equation}

Further identities can be derived using the Dirac matrices technology, namely, in Euclidean space we have

\begin{eqnarray}
\gamma_\mu =
\begin{pmatrix}
0&-i\sigma_\mu\\
i\overline{\sigma}_\mu&0
\end{pmatrix}, \quad
\gamma_5 =
\begin{pmatrix}
I&0\\
0&-I
\end{pmatrix}
\end{eqnarray}

\noindent then we can compute the trace of four Pauli matrices using the identities for Dirac matrices \cite{Peskin}, and projecting out the lower block with
the chiral projector $(1-\gamma_5)/2$. We have for example

\begin{eqnarray}\nonumber
{\rm Tr}(\sigma_\mu\overline{\sigma}_\nu\sigma_\rho\overline{\sigma}_\lambda)=\frac{1}{2}{\rm Tr}(\gamma_\mu\gamma_\nu\gamma_\rho\gamma_\lambda(1-\gamma^5))
=2(\delta_{\mu\nu}\delta_{\rho\lambda}+\delta_{\mu\lambda}\delta_{\rho\nu}-\delta_{\mu\rho}\delta_{\nu\lambda}+\epsilon_{\mu\nu\rho\lambda})
\end{eqnarray}

\subsection{General Results}\label{result}

\subsubsection{Block separation $L\geq 1$}

In section \ref{N2systems} we found an explicit expression for the density matrix of two blocks of length $L_1$ and $L_2$ separated by 
$L$ sites. In that section we presented the asymptotic results for the eigenvalues of $\rho_{AB}$. The result for any $L_1,L_2\geq1$ 
is given in terms of the following quantities

\begin{flalign}\nonumber
\lambda_{0}(L)=\frac{1}{4}\left(1 +3 \left(-\frac{1}{3}\right)^{L}\right),\quad
\lambda_{1}(L)=\frac{1}{4}\left(1 -\left(-\frac{1}{3}\right)^{L}\right).
\end{flalign}

\noindent From those quantities we define $\lambda_{00}=\lambda_{0}(L_1)\lambda_{0}(L_2)$,
$\lambda_{11}=\lambda_{1}(L_1)\lambda_{1}(L_2)$ and 
$\lambda_{10}=\lambda_{0}(L_1)\lambda_{1}(L_2)+\lambda_{0}(L_2)\lambda_{1}(L_1)$.

The characteristic polynomial associated to the density matrix $\rho_{AB}$, $p(Y)=\det(\rho_{AB}-Y)$ is
$p(Y)=p_1(Y)^5p_2(Y)p_3(Y)^3,$  where $p_k(Y)$ is a polynomial of degree $k$ on $Y$, given by ($z=(-3)^{-L}$)

\begin{eqnarray}\nonumber
p_1(Y)&=&Y-(1-z)\lambda_{11},\\\nonumber
p_2(Y)&=&Y^2-(\lambda_{00}+(1+2z)\lambda_{11})Y+(1-z)(1+3z)\lambda_{00}\lambda_{11},\\\nonumber
p_3(Y)&=&Y^3-(\lambda_{10}+\lambda_{11}(1+z))Y^2+[(1+z)\lambda_{00}+(1+2z)\lambda_{10}](1-z)\lambda_{11}Y\\
&-&(1-z)^2(1+3z)\lambda_{00}\lambda_{11}^2=Y^3+bY^2+cY+d.\label{poly}
\end{eqnarray}

We define $q\equiv\frac{1}{27}(2b^3-9bc+27d),\quad p\equiv\frac{1}{3}(3c-b^2)$. The eigenvalues of $\rho_{AB}$ are the solutions to 
$P(y)=0$. They are 

\begin{eqnarray}
y&=&(1-z)\lambda_{11}\quad \mbox{five-fold degeneracy},\\
y&=&\frac{1}{2}\left((\lambda_{00}+(1+2z)\lambda_{11})\pm\sqrt{(\lambda_{00}+(1+2z)\lambda_{11})^2-4(1-z)(1+3z)\lambda_{00}\lambda_{11}}\right),\\
y&=& 2\sqrt {-\frac{p}{3}}\cos \left( \frac{1}{3}\arccos \left( \frac{3q}{2p}\sqrt {-\frac{3}{p}}\right) +\frac{2\pi k}{3}\right) -\frac{b}{3},\,\,\mbox{($k=0,1,2$) triple deg.}
\end{eqnarray}

\subsubsection{Adjacent blocks}
In section \ref{N2systems}, we computed the transposed density matrix of a system consisting of two blocks inside the VBS state, $A$ and $B$, 
of length $L_1$ and $L_2$ respectively. The spins which do not belong to $A\cup B$ have been traced away. In the general case when the blocks 
are separated by $L$ sites we could prove that the negativity vanishes for $L\geq1$, being the only nontrivial case when $L=0$. In that case, 
the negativity in the asymptotic limit $L_1\rightarrow\infty,L_2\rightarrow\infty$ is given by

\begin{equation}
 \mathcal{N}_{L_1,L_2\rightarrow\infty}=\frac{1}{2}-\frac{3}{4}\left(\left(-\frac{1}{3}\right)^{2L_1}+\left(-\frac{1}{3}\right)^{2L_2}\right).
\end{equation}

The decay in the thermodynamic limit is twice as fast compared to the usual decay of the spin correlations, a feature that is already seen in the 
case of the negativity of the pure system studied before.

The logarithmic negativity also show this behavior, for $L_1,L_2\gg1$

\begin{equation}
 E_\mathcal{N}=1-\frac{3}{4\ln(2)}\left(\left(-\frac{1}{3}\right)^{2L_1}+\left(-\frac{1}{3}\right)^{2L_2}\right).
\end{equation}

In the of adjacent blocks, for any $L_1,L_2\geq1$, the characteristic polynomial associated to the transposed density 
matrix $\rho^{T_A}_{AB}$, $p(y)=\det(\rho_{AB}^{T_{A}}-Iy)$ is $\bar{p}(y)=\bar{p}_1(y)^5\bar{p}_2(y)\bar{p}_3(y)^3,$  where $\bar{p}_k(y)$ is a polynomial of 
degree $k$ on $y$, given by (defined in terms of $\lambda_{00},\lambda_{10},\lambda_{11}$)

\begin{flalign}\nonumber
&\bar{p}_1(y)=y-2\lambda_{11}(1,2),&\\\nonumber
&\bar{p}_2(y)=y^2-y(\lambda_{00}(1,2)-\lambda_{11}(1,2))-4\lambda_{00}(1,2)\lambda_{11}(1,2),&\\
&\bar{p}_3(y)=y^3-\lambda_{10}(1,2)y^2-2\lambda_{10}(1,2)\lambda_{11}(1,2)y+8\lambda_{00}(1,2)\lambda_{11}(1,2)^2=y^3+by^2+cy+d.&
\end{flalign}

\noindent this polynomials are related with the polynomials of previous section by taking $z=-1$ in eq. (\ref{poly}).
Out of the 16 eigenvalues ($y_n$), 4 are negative, with $y_1$ (no degeneracy) given by the expression

\begin{eqnarray}
y_1= \frac{1}{2}\bigg(\lambda_{00}(1,2)-\lambda_{11}(1,2)
-\sqrt{\lambda_{00}(1,2)^2+14\lambda_{00}(1,2)\lambda_{11}(1,2)+\lambda_{11}(1,2)^2}\bigg)
\end{eqnarray}

\noindent and $y_2$ (triple degeneracy) given by

\begin{equation}
y_2= -2\sqrt {-\frac{p}{3}}\sin \left( \frac{1}{3}\arccos \left( \frac{3q}{2p}\sqrt {-\frac{3}{p}}\right) +\frac{\pi}{6}\right) -\frac{b}{3}
\end{equation}

The negativity of the system is then $\mathcal{N}=-(y_1+3y_2)$, while the logarithmic negativity is given by $E_\mathcal{N}=\log_2(1-2(y_1+3y_2)).$ 
In the special case when $L_1=L_2=l$, the negativity simplifies to (using $x=(-3)^{-l}$)

\begin{eqnarray}
\mathcal{N}(l)&=&-\frac{1}{4}\bigg(x+{x}^{2}-\frac{1}{2}\sqrt {1+4x+2{x}^{2}-4{x}^{3}+13{x}^{4}}
-\frac{3}{4}\sqrt {(1+3x)(1-x)^{3}}\bigg),\\
&\simeq& \frac{1}{2}-\frac{3}{2}(x^2-x^3), \quad \mbox{for $x$}\ll 1.
\end{eqnarray}

The logarithmic negativity is given by $E_\mathcal{N}(\rho_{AB})\simeq 1-\frac{3}{2\ln(2)}(x^2-x^3).$

\end{document}